\documentclass{amsart}
\usepackage{amsmath,amsthm}
\usepackage{amsfonts,amssymb}
\usepackage{mathrsfs}

\usepackage{enumerate}

\newtheorem{theorem}{Theorem}

\newtheorem{lemma}[theorem]{Lemma}

\numberwithin{theorem}{section}

\theoremstyle{definition}

\theoremstyle{remark}
\newtheorem{remark}[theorem]{Remark}

\newcommand{\mR}{\mathbb{R}}

\newcommand{\mS}{\mathbb{S}}

\newcommand{\cP}{\mathcal{P}}

\newcommand{\cH}{\mathcal{H}}

\newcommand{\cL}{\mathcal{L}}
\newcommand{\cR}{\mathcal{R}}
\newcommand{\cU}{\mathcal{U}}

\newcommand{\su}{\mathfrak{su}(1,1)}

\catcode`,\active

\catcode`\,12

\newcommand*\pFqskip{8mu}
\catcode`,\active
\newcommand*\pFq{\begingroup
        \catcode`\,\active
        \def ,{\mskip\pFqskip\relax}%
        \dopFq
}
\catcode`\,12
\def\dopFq#1#2#3#4#5{%
        {}_{#1}F_{#2}\biggl[\genfrac..{0pt}{}{#3}{#4};#5\biggr]%
        \endgroup
}

\begin{document}
 
\title
{Bargmann and Barut-Girardello models for the Racah algebra}
\author[H. De Bie]{Hendrik De Bie}
\address{Department of Mathematical Analysis, Faculty of Engineering and Architecture, Ghent University, Krijgslaan 281, 9000 Gent, Belgium.}
\email{Hendrik.DeBie@UGent.be}

\author[P. Iliev]{Plamen Iliev}
\address{School of Mathematics, Georgia Institute of Technology, Atlanta, GA 30332-0160, USA.}
\email{iliev@math.gatech.edu}

\author[L. Vinet]{Luc Vinet}
\address{Centre de Recherches Math\'ematiques, Universit\'e de Montr\'eal, P.O. Box 6128, Centre-ville Station, Montr\'eal, QC H3C 3J7, Canada}
\email{vinet@crm.umontreal.ca}

\date{\today}

\begin{abstract}

The Racah algebra and its higher rank extension are the algebras underlying the univariate and multivariate Racah polynomials. In this paper we develop two new models in which the Racah algebra naturally arises as symmetry algebra, namely the Bargmann model and the Barut-Girardello model. We show how both models are connected with the superintegrable model of Miller et al.
The Bargmann model moreover leads to a new realization of the Racah algebra of rank $n$ as $n$-variable differential operators. Our conceptual approach also allows us to rederive the basis functions of the superintegrable model without resorting to separation of variables.
\end{abstract}

\maketitle


\section{Introduction}
\setcounter{equation}{0}

\subsection{Racah algebra and Racah polynomials}
The Racah algebra was first introduced in \cite{Granovskii&Zhedanov-1988}. It is the infinite-dimensional associative algebra over $\mathbb{C}$ with generators $K_1$ and $K_2$ that satisfy, together with their commutator $K_3=\left[K_1,K_2\right]$, the relations:
\begin{align}\label{eq:Racah}
\begin{split}
\left[K_2,K_3\right]&=K_2^2+\left\{K_1,K_2\right\}+dK_2+e_1, \\
\left[K_3,K_1\right]&=K_1^2+\left\{K_1,K_2\right\}+dK_1+e_2,
\end{split}
\end{align}
with $\left\{A,B\right\}:=AB+BA$ and where $d$, $e_1$, $e_2$ are structure constants. It is crucially related with the Racah polynomials, which sit at the top of the discrete branch of the celebrated Askey scheme \cite{Koek}. Indeed, a suitable realization of the Racah algebra in terms of difference operators encodes the bispectral properties of these polynomials.

The Racah algebra can alternatively be presented in terms of generators  $C_{12}$, $C_{23}$, $C_{13}$ and $F$. In this form one has $C_{12}+C_{23}+C_{13}=G$ and $\left[C_{23},C_{13}\right]=\left[C_{13},C_{12}\right]=\left[C_{12},C_{23}\right]=2F$, where $G$ is a central operator. The relations then are 
\begin{align} \label{equit}
\begin{split}
\left[C_{12},F\right]&=C_{23}C_{12}-C_{12}C_{13}+i_{12},\\
\left[C_{23},F\right]&=C_{13}C_{23}-C_{23}C_{12}+i_{23}, \\
\left[C_{13},F\right]&=C_{12}C_{13}-C_{13}C_{23}+i_{13},
\end{split}
\end{align}
with $i_{23}, i_{13}, i_{12}$ central elements. The formulas (\ref{equit}) were first obtained in \cite{Gao&Wang&Hou-2013}. It is precisely this presentation we will use in this paper.
 
Multivariate extensions of the Racah polynomials were introduced by Tratnik in \cite{Trat}. These polynomials are multispectral, as obtained through an elaborate construction in \cite{geronimo}. 

Also the Racah algebra can be generalized to higher rank. This was first achieved in \cite{I} in the framework of superintegrability. Indeed, it appears naturally by considering the Hamiltonian 
\begin{equation} \label{eq:Hamiltonian1}
 H:=\sum_{1\leq i < j \leq n+1} (y_j \partial_{y_i} - y_i \partial_{y_j})^2
 +\sum_{i=1}^{n+1}\frac{b_i}{y_i^2},
\end{equation}
which is maximally superintegrable, and was introduced by Miller et al. \cite{KMP1, KMP2, K1, MPW, MT}. The symmetry algebra of this Hamiltonian is precisely the higher rank Racah algebra \cite{I} and it acts irreducibly on the eigenspaces of the Hamiltonian \cite{I2}. 
In \cite{racah} the higher rank Racah algebra is constructed in an alternative way, namely within the $n$-fold tensor product of $\cU(\su)$. The construction goes as follows. Consider the Casimir $C$ of $\su$. Using the coproduct, $C$ can be lifted to an arbitrary number of factors in the tensor product. This defines  intermediate Casimirs $C_A$ for any subset $A \subset [n]=\{1, \ldots, n\}$, where $A$ describes the relevant factors in the tensor product. The algebra generated by $\{ C_A\}$ is then called the Racah algebra of rank $n-2$. In the present paper we will show that both approaches coincide.

Recently, it was shown that the higher rank Racah algebra also has a discrete realization acting directly on the multivariate Racah polynomials as difference operators, see \cite{wouter}. This action extends the diagonal action already obtained in \cite{geronimo} and completely describes the connection between the Racah polynomials and their algebra.

Finally note that the Racah algebra can also be considered from the point of view of Howe duality, see \cite{Howe, Howe2}.

\subsection{Two new models}
The operator $\sum_{j=1}^n \partial_j$, while deceptively simple, has a rather rich history of investigations. First, it has been investigated in the case where it acts on polynomials $f(x_1, x_2, \ldots, x_n)$ with $x_j$ either $0$ or $1$. This was initially introduced by Dunkl in \cite{Du1, Du2} and simultaneously by Delsarte in \cite{Del}.
For a nice overview from the point of view of harmonic analysis, see \cite{S}. More recently, an orthogonal basis for null-solutions in this context was obtained in \cite{Film}.

Second, it has also been investigated in the case where it acts on complex-valued functions over $\mR^n$, which is also the topic of the present paper. The two relevant references are \cite{H1} and \cite{H2}, where Rosengren uses this operator to extend the theory of Hankel forms to the multilinear case. In doing so, he develops partly the harmonic analysis for this operator (such as projection operators). Nevertheless, he does not look into the problem of the symmetries of this operator.

In the present paper, we show that the operator  $\sum_{j=1}^n \partial_j$ arises as the $n$-fold tensor product of $\su$ by considering the so-called Bargmann realization. This means that its symmetries give rise to a new realization of the higher rank Racah algebra. We construct a basis of null-solutions that explicitly diagonalizes an abelian subalgebra of the Racah algebra. Moreover, by suitable acting on this basis we are able to construct a realization of the rank $n$ Racah algebra in precisely $n$ variables.

We also consider an alternative realization of $\su$, called the Barut-Girardello realization \cite{BG}. This gives rise to a second new model using the differential operator $\sum_{j=1}^n \left(x_j\partial_{x_{j}}^2 + 2 \nu_j \partial_{x_j} \right)$. Here the $\nu_j$ are constants. Surprisingly, both models turn out to be isomorphic. This is achieved by using a modified Laplace transform, first introduced in \cite{BVM}.
Finally, we show that the Barut-Girardello model can be identified with the Miller model. This identification shows precisely how an $n$-fold tensor product of $\su$ is underlying the Miller model. Moreover, it yields a new way of obtaining its basis functions (see e.g. \cite{I, I2}), using Fischer decomposition and Cauchy-Kovalevskaia extension rather than separation of variables.

\subsection{Contributions and organization}
In conclusion, our main contributions in this paper are the following:
\begin{itemize}
\item A realization of the rank $n$ Racah algebra in precisely $n$ variables.
\item The isomorphism between the Bargmann, Barut-Girardello and Miller models.
\item The conceptual approach to find the basis functions used in \cite{I, I2}.
\end{itemize}

The paper is organized as follows. In Section 2 we recall the general construction of the Racah algebra. In Section 3 we develop the Bargmann realization of the Racah algebra. In Section 4 we introduce the Barut-Girardello model and show how it is isomorphic with the Bargmann model through a modified Laplace transform. Finally, in Section 5 we show that the Barut-Girardello model can be identified with the superintegrable model of Miller et al.

\section{Preliminaries}
\setcounter{equation}{0}
\label{prelims}

The Lie algebra $\mathfrak{su}(1,1)$ is generated by $J_\pm$ and $A_0$ with relations:
\begin{equation}\label{su11}
[J_-,J_+]=2A_0, \qquad [A_0, J_\pm]=\pm J_\pm.
\end{equation}
Its universal enveloping algebra $\mathcal{U}(\mathfrak{su}(1,1))$ contains the Casimir element of $\mathfrak{su}(1,1)$:
\begin{equation}\label{Casimir}
	C:=A_0^2-A_0-J_+J_-.
\end{equation}
The comultiplication $\mu^*$ maps $\mathfrak{su}(1,1)$ into $\mathfrak{su}(1,1)\otimes\mathfrak{su}(1,1)$ via:
\begin{align*}
 \mu^*(J_{\pm})=J_{\pm}\otimes 1 + 1 \otimes J_{\pm}, \\
  \mu^*(A_{0})=A_{0}\otimes 1 + 1 \otimes A_{0}.
\end{align*}
This map extends to $\mathcal{U}(\mathfrak{su}(1,1))$. 

We now define inductively
$$
\mathscr{C}_1:=C, \qquad  \mathscr{C}_n:=(\underbrace{1\otimes\ldots\otimes 1}_{n-2 \text{ times }} \otimes \mu^*)(\mathscr{C}_{n-1}).
$$
 Consider the map
$$
\tau_k: \bigotimes_{i=1}^{m-1} \mathcal{U}(\mathfrak{su}(1,1))\rightarrow \bigotimes_{i=1}^{m} \mathcal{U}(\mathfrak{su}(1,1)),
$$
acting on homogeneous tensor products as:
$$
\tau_k(t_1\otimes \ldots \otimes t_{m-1}):=t_1\otimes \ldots \otimes t_{k-1} \otimes 1 \otimes t_k \otimes \ldots \otimes t_{m-1} .
$$
and extend it by linearity. We then put
\begin{align}
\label{Casimir-Upper}
C_A:=\left(\prod_{k \in \left[ n\right] \backslash A}^{\longrightarrow} \tau_k \right)\left(\mathscr{C}_{|A|}\right),
\end{align}
with $A$ a subset of $[n]$. Here we use the notation $[n]:=\{1,\ldots, n \}$.


The subalgebra $\cR_n$ of $\bigotimes_{i=1}^{n} \mathcal{U}(\mathfrak{su}(1,1))$ generated by the set $\{ C_A | A \subset [n] \}$ is called the Racah algebra.
For more details we refer the reader to \cite{racah}. We summarize the facts we will need in the sequel.
\begin{lemma}\label{Commute}
 If either $A \subset B$ or $B \subset A$ or $A \cap B=\emptyset$ then $C_A$ and $C_B$ commute.
\end{lemma}
As a consequence $C_i$ and $C_{[n]}$ are central in $\cR_n$. Consider a chain of sets $A_1 \varsubsetneq A_2 \varsubsetneq \dots A_j$, each set strictly contained in the next one. The operators $C_{A_i}$ corresponding to these sets  generate an Abelian algebra. 

Now let $\mathcal{C}$ be a chain  with maximal length: $A_i \varsubsetneq A_{i+1}$ with $|A_{i+1}|=|A_{i}|+1$. The set of Casimirs belonging to this chain excluding the central Casimirs generate the Abelian algebra $\langle C_{B}|B \in \mathcal{C} \text { and } |B|\neq 1, n \rangle$. We call this a labelling Abelian algebra. The rank of a Racah algebra equals the dimension of a labelling Abelian algebra, so $\cR_n$ has rank $n-2$.

Finally, the main relations of the Racah algebra are given in the following: \begin{lemma}\label{Racah1}
Let $K$, $L$ and $M$ be three disjoint subsets of $[n]$. The algebra generated by the set
\[ \{ C_K,C_L,C_M,C_{K \cup L}, C_{K \cup M}, C_{ L \cup M}, C_{K \cup L \cup M}  \} \]
is isomorphic to the rank $1$ algebra $\cR_3$. Therefore, putting $C_{KL} =C_{K\cup L}$ and introducing the operator $F$ as:
\begin{align}\label{B}
2F:=[C_{KL},C_{LM}] = [C_{KM},C_{KL}] = [C_{LM},C_{KM}],
\end{align}
the following relations hold 
\begin{align} \label{RankoneRacah}
\begin{split}
[C_{KL},F]&=C_{LM}C_{KL}-C_{KL}C_{KM}+\left(C_L-C_K\right)\left(C_{M}-C_{KLM}\right), \\
[C_{LM},F]&=C_{KM}C_{LM}-C_{LM}C_{KL}+\left(C_M-C_L\right)\left(C_{K}-C_{KLM}\right),\\
[C_{KM},F]&=C_{KL}C_{KM}-C_{KM}C_{LM}+\left(C_K-C_M\right)\left(C_{L}-C_{KLM}\right).
\end{split}
\end{align}
\end{lemma}

\section{The Bargmann model}
\setcounter{equation}{0}

\subsection{Construction of the model}

Introduce, for $\nu >0$, the following operators
\begin{align*}
K_0 &= x \partial_x + \nu\\
K_- &= \partial_x\\
K_+ &= x^2 \partial_x + 2 \nu x.
\end{align*}
It is easy to verify that they satisfy the $\su$ relations:
\[
[K_0, K_\pm] = \pm K_\pm, \qquad [K_-,K_+] = 2 K_0.
\]

Consider now $n$ mutually commuting sets of such $\su$ generators and combine them by addition. This yields the following operators
\begin{align*}
K_-^{[n]} &= \sum_{j=1}^n\partial_{x_{j}}\\
K_+^{[n]} &= \sum_{j=1}^n(x_j^2 \partial_{x_{j}} + 2 \nu_j x_j)\\
K_0^{[n]} & = \sum_{j=1}^n x_j\partial_{x_{j}} + \sum_{j=1}^n\nu_j
\end{align*}
which again generate $\su$. 

We consider their action on $\cP(\mR^n)= \mR[x_1, \ldots, x_n]$. 
We define the space of harmonics
\[
\cH_k(\mR^n) = \cP_k(\mR^n) \cap \ker{K_-^{[n]}}
\]
where $\cP_k(\mR^n)$ is the space of homogeneous polynomials of degree $k$.

Then we have two fundamental results. First is the Fischer decomposition of homogeneous polynomials into harmonics:
\begin{theorem}
\label{Fischer}
The space $\cP_k(\mR^n)$ decomposes as
\[
\cP_k(\mR^n) = \bigoplus_{j=0}^k \left(K_+^{[n]}\right)^j \cH_{k-j}(\mR^n).
\]
\end{theorem}

\begin{proof}
This is standard and is an essential consequence of the $\su$ relations. One can easily e.g. adapt the proof of Theorem 3 in \cite{sommen} to work for the present situation.
The statement is also a consequence of Section 8 in \cite{H2}.
\end{proof}

Second we have the Cauchy-Kovalevskaia (CK) isomorphism that yields a simple description of all harmonics.
\begin{theorem}
\label{CK}
The space $\cH_k(\mR^n)$ is isomorphic with $\cP_k(\mR^{n-1})$. The isomorphism is given by the map
\begin{align*}
CK_{n} : \cP_k(\mR^{n-1}) &\rightarrow \cH_k(\mR^{n})\\
p(x_1, \ldots, x_{n-1})&\mapsto p(x_1-x_n, \ldots, x_{n-1}-x_n)
\end{align*}
and its inverse by putting $x_n =0$.
\end{theorem}

\begin{proof}
Let us write
\[
CK_n \left( p (x_1, \ldots, x_{n-1})\right) = \sum_{j=0}^k x_n^j p_j(x_1, \ldots, x_{n-1})
\]
with $p_0 :=p$ and $p_j$, $j  \in \{1, \ldots, k\}$ to be determined. We express that $K_-^{[n]} CK_n \left( p \right) =0$. This yields
\begin{align*}
&K_-^{[n]} CK_n \left( p (x_1, \ldots, x_n)\right) =0\\
\Longleftrightarrow & \left( K_-^{[n-1]} + \partial_{x_n} \right) CK_n \left( p (x_1, \ldots, x_n)\right) =0\\
\Longleftrightarrow &x_n^k K_-^{[n-1]} p_k + \sum_{j=0}^{k-1} x^j_n \left( (j+1) p_{j+1} + K_-^{[n-1]} p_j \right)  =0.
\end{align*}
This yields, for all $j  \in \{1, \ldots, k\}$
\[
p_j = \frac{(-1)^j}{j !} \left( K_-^{[n-1]} \right)^j p.
\]
Hence we have
\begin{align*}
CK_n p &= \sum_{j=0}^k \frac{(- x_n)^j \left( K_-^{[n-1]} \right)^j}{j!} p\\
&= \exp(- x_n K_-^{[n-1]}) p\\
& = \prod_{j=1}^{n-1} \exp(-x_n \partial_{x_j}) p\\
&=p(x_1-x_n, \ldots, x_{n-1}-x_n)
\end{align*}
as $\exp(-x_n \partial_{x_j})$ is the translation operator.

Using these explicit formulas, it is now easily seen that $CK_n$ is indeed an isomorphism of vector spaces. Indeed, $CK_n$ is injective by construction, and so is its left inverse $R$ which is given by $R q(x_1, \ldots, x_n) = q(x_1, \ldots, x_{n-1}, 0)$.
\end{proof}

We can now obtain a basis for $\cH_k(\mR^{n})$, by combining Theorem \ref{CK} with Theorem \ref{Fischer}. We have the following result
\begin{theorem}
\label{basis}
A basis for $\cH_k(\mR^{n})$ is given by the set of polynomials
\[
\psi_{j_1, \ldots, j_{n-1}} =  CK_{n} \left(  \left(K_+^{[n-1]}\right)^{j_{n-1}} CK_{n-1} \left(  \ldots \left(K_+^{[3]}\right)^{j_3} CK_{3} \left( \left(K_+^{[2]}\right)^{j_2} CK_{2} \left( x_1^{j_1}   \right) \right) \right) \right)
\]
with $\sum_{\ell=1}^{n-1} j_{\ell} = k$.
\end{theorem}

This basis is special in the following sense. Consider a subset $B \subset [n]$ and introduce the operators
\begin{align*}
K_-^{B} &= \sum_{j \in B}\partial_{x_{j}}\\
K_+^{B} &= \sum_{j \in B}(x_j^2 \partial_{x_{j}} + 2 \nu_j x_j)\\
K_0^{B} & = \sum_{j \in B} x_j\partial_{x_{j}} + \sum_{j \in B}\nu_j
\end{align*}
which again yields an `intermediate' $\su$ realization. Its Casimir is given by
\begin{equation}
\label{casbarg}
C_B:=\left(K_0^B\right)^2 - K_0^B - K_+^B K_-^B.
\end{equation}
The collection of Casimir operators $C_B$, for all $B \subset [n]$, generate the algebra $\cR_n$, which is called the rank $(n-2)$ Racah algebra, see Section \ref{prelims}.
In particular, we consider the following labelling Abelian subalgebra:
\[
\langle C_{[2]}, C_{[3]}, C_{[4]}, \ldots, C_{[n-1]}, C_{[n]}  \rangle
\]
and show that it acts diagonally on the basis of Theorem \ref{basis}. Indeed, we have
\begin{theorem}
\label{quantumnrs}
One has, for $\ell \in \{2, \ldots, n\}$
\[
C_{[\ell]} \psi_{j_1, \ldots, j_{n-1}} = \lambda_{j_1, \ldots, j_{n-1}}^{[\ell]}\psi_{j_1, \ldots, j_{n-1}}
\]
with
\[
\lambda_{j_1, \ldots, j_{n-1}}^{[\ell]}= \left(\sum_{i=1}^{\ell-1} j_i + \sum_{i=1}^{\ell} \nu_i \right) \left(-1+\sum_{i=1}^{\ell-1} j_i + \sum_{i=1}^{\ell} \nu_i \right).
\]
\end{theorem}

\begin{proof}
We first observe that if $\ell$ equals the dimension of our space, we have
\begin{align}
\label{eig}
\begin{split}
C_{[\ell]} \cH_k(\mR^\ell) &= K_0^{[\ell]}\left(  K_0^{[\ell]} -1 \right)\cH_k(\mR^\ell)\\
&= \left(k + \sum_{i=1}^{\ell} \nu_i \right)\left(k -1+ \sum_{i=1}^{\ell} \nu_i \right)\cH_k(\mR^\ell).
\end{split}
\end{align}
Subsequently, as $C_{[\ell]}$ commutes by construction with $CK_p$ for $p > \ell$ and with $K_+^{[p]}$ for $p \geq \ell$, we have
\begin{align*}
C_{[\ell]} \psi_{j_1, \ldots, j_{n-1}} &= CK_{n} \left(  \left(K_+^{[n-1]}\right)^{j_{n-1}} CK_{n-1} \left(  \ldots\right. \right.\\
&\quad \ldots \left(K_+^{[\ell]}  \right)^{j_\ell} C_{[\ell]} \left( CK_{\ell} \left( \ldots   CK_{3} \left( \left(K_+^{[2]}\right)^{j_2} CK_{2} \left( x_1^{j_1}   \right) \right) \right) \right).
\end{align*}
The result then follows by application of (\ref{eig}).
\end{proof}

\subsection{Description of the Racah algebra in the new variables}
\label{SecVars}

The Racah algebra $\cR_n$ is generated by the operators $C_B$ of formula (\ref{casbarg}) which are defined using $n$ variables. However, the algebra is only of rank $n-2$. In this section, we show how a realization of $\cR_n$ can be constructed using exactly $n-2$ variables.

Consider the space $\cH_k(\mR^n)$. Using Theorem \ref{CK} we find that an alternative basis for $\cH_k(\mR^n)$ is given by
\begin{align*}
\varphi_{j_1, \ldots, j_{n-2}} & = (x_1 - x_2)^{k- j_1 -j_2 - \ldots -j_{n-2}} (x_3-x_2)^{j_1} (x_4-x_3)^{j_2} \ldots (x_{n}-x_{n-1})^{j_{n-2}} \\
&= (x_1 - x_2)^k u_1^{j_1} u_2^{j_2} \ldots u_{n-2}^{j_{n-2}}
\end{align*}
with $j_\ell$ positive integers with $\sum_{\ell=1}^{n-2} j_\ell \leq k$.
Here we introduced $n-2$ new variables $\{ u_1, u_2, \ldots, u_{n-2} \}$ given by
\[
u_j := \frac{x_{j+2} - x_{j+1}}{x_1 - x_2}, \qquad j \in \{ 1, \ldots, n-2 \}.
\]

The action of $\cR_n$ on $\cH_k(\mR^n)$, through the use of this basis, is hence projected to the space
\[
\Pi_k^{n-2}= \oplus_{\ell=0}^k \cP_\ell(u_1, \ldots, u_{n-2}),
\] 
i.e. the space of polynomials of total degree at most $k$ in $n-2$ variables.

Gauging the operators $C_B$ by $(x_1 - x_2)^k$ to
\begin{equation}
\label{gauge}
\widetilde{C_B} = (x_1 - x_2)^{-k} C_B (x_1 - x_2)^{k}
\end{equation}
then yields a realization of $\cR_n$ on $\Pi_k^{n-2}$. We give this realization explicitly in the following theorem.

\begin{theorem}
The space $\Pi_k^{n-2}$ of all polynomials of degree $k$ in $n-2$ variables carries a realization of the rank $n-2$ Racah algebra $\cR_n$.
This realization is given explicitly by
\[
\widetilde{C_{i}}=\nu_i(\nu_i-1), \qquad i \in [n]
\]
and, for $i, j \in \{3, \ldots, n \}$,
\begin{align*}
\widetilde{C_{12}}&=- \left(k-1-\sum_{\ell=1}^{n-2} u_{\ell}\partial_{u_\ell} \right)  \left(-k-\partial_{u_1}+\sum_{\ell=1}^{n-2} u_{\ell}\partial_{u_\ell} \right) + 2 \nu_2 \left(k-\sum_{\ell=1}^{n-2} u_{\ell}\partial_{u_\ell} \right) \\
& \qquad - 2 \nu_1\left(-k-\partial_{u_1}+\sum_{\ell=1}^{n-2} u_{\ell}\partial_{u_\ell} \right) + (\nu_1+\nu_2)(\nu_1+\nu_2-1)\\
\widetilde{C_{1j}}&=- \left(1 - \sum_{\ell=1}^{j-2} u_{\ell} \right)^2 \left(k-1-\sum_{\ell=1}^{n-2} u_{\ell}\partial_{u_\ell} \right)   \left( \partial_{u_{j-2}}- \partial_{u_{j-1}} \right)\\
& \qquad + 2 \nu_j \left(1 - \sum_{\ell=1}^{j-2} u_{\ell} \right)\left(k-\sum_{\ell=1}^{n-2} u_{\ell}\partial_{u_\ell} \right) 
 - 2 \nu_1  \left(1 - \sum_{\ell=1}^{j-2} u_{\ell} \right)  \left( \partial_{u_{j-2}}- \partial_{u_{j-1}} \right)\\ & \qquad  + (\nu_1+\nu_j)(\nu_1+\nu_j-1)\\
\widetilde{C_{2j}}&= -\left( \sum_{\ell=1}^{j-2} u_{\ell} \right)^2 \left(1-k-\partial_{u_1}+\sum_{\ell=1}^{n-2} u_{\ell}\partial_{u_\ell} \right)   \left( \partial_{u_{j-2}}- \partial_{u_{j-1}} \right)\\
& \qquad + 2 \nu_j \left( \sum_{\ell=1}^{j-2} u_{\ell} \right) \left(k+\partial_{u_1}-\sum_{\ell=1}^{n-2} u_{\ell}\partial_{u_\ell} \right)  + 2 \nu_2 \left( \sum_{\ell=1}^{j-2} u_{\ell} \right)\left( \partial_{u_{j-2}}- \partial_{u_{j-1}} \right)   \\
& \qquad + (\nu_2+\nu_j)(\nu_2+\nu_j-1)\\
\widetilde{C_{ij}}&= - \left( \sum_{\ell=j-1}^{i-2} u_{\ell} \right)^2 \left( \partial_{u_{i-2}}- \partial_{u_{i-1}} \right) \left( \partial_{u_{j-2}}- \partial_{u_{j-1}} \right)\\
 & \qquad+ 2 \nu_j  \left( \sum_{\ell=j-1}^{i-2} u_{\ell} \right)  \left( \partial_{u_{i-2}}- \partial_{u_{i-1}} \right) - 2 \nu_i  \left( \sum_{\ell=j-1}^{i-2} u_{\ell} \right)  \left( \partial_{u_{j-2}}- \partial_{u_{j-1}} \right)\\
& \qquad + (\nu_i+\nu_j)(\nu_i+\nu_j-1)
\end{align*}
where we assume $i >j$ and with $u_{n-1}=0$ whenever it appears.
\end{theorem}

\begin{proof}
In \cite{racah} it is shown that for any $A \subset [n]$
\[
C_A = \sum_{\{i,j\} \subset A} C_{ij} - \left( |A|-2\right) \sum_{i \in A} C_i.
\]
Therefore it suffices to find the expressions for $\widetilde{C_{i}}$ and $\widetilde{C_{ij}}$. 
Observe that from (\ref{casbarg}) follows
\[
C_i = \nu_i(\nu_i-1)
\]
and
\begin{align*}
C_{ij} &= - (x_i-x_j)^2 \partial_{x_i} \partial_{x_j} + 2 \nu_j (x_i - x_j) \partial_{x_i} - 2 \nu_i (x_i - x_j) \partial_{x_j}\\
& \quad + (\nu_i + \nu_j) (\nu_i + \nu_j-1).
\end{align*}
Combining these formulas with (\ref{gauge}), through long and tedious computations, yields the formulas of the theorem.
\end{proof}

\subsection{Explicit formulas for basisfunctions}

In \cite{equit}, the special case $n=3$ was studied and a basis for $\cH_k(\mR^3)$ was given explicitly in terms of hypergeometric functions. We show how this can be rederived from our Theorem \ref{basis}. Indeed, according to Theorem \ref{basis} the basis is given as
\begin{equation}
\label{3deig}
\psi_{j}^k =   CK_{3} \left( \left(K_+^{[2]}\right)^{k-j} CK_{2} \left( x_1^{j}   \right) \right),
\end{equation}
with $j \in \{0, \ldots, k\}$.
We readily observe that, using Theorem \ref{CK}
\[
CK_{2} \left( x_1^{j}   \right) = (x_1 - x_2)^j.
\]
Because of Theorem \ref{CK} we may also write
\[
\psi_{j}^k(x_1, x_2, x_3) = (x_1-x_2)^k \phi^k_{k-j}(u), \qquad u = \frac{x_3-x_2}{x_1-x_2}.
\]
We derive a recursive relation for $\phi_j^k$. From (\ref{3deig}) we have
\[
\psi_j^{k+1} = \left( CK_3 K_+^{[2]} CK_3^{-1}  \right) \psi_j^{k} 
\]
which, after some computations, translates to
\begin{equation}
\label{recursiveL}
\phi^{k+1}_{k+1-j} = \left(u (u-1) \frac{d}{du} + (2 \nu_1+k)-(2 \nu_1+2 \nu_2+2 k)u  \right)\phi_{k-j}^k
\end{equation}
with initial condition $\phi^k_0 = 1$. 


The hypergeometric series enjoys the following property:
\[
(1-c)   \pFq{2}{1}{a-1,b-1}{c-1}{u}= \left(u(u-1) \frac{d}{du} + 1-c + (a+b-1)u \right)\pFq{2}{1}{a,,b}{c}{u}.
\]
This can be found for instance by combining formulas (2.5.1) and (2.5.2) in \cite{AAR}. Comparing this result with the recursive relation (\ref{recursiveL}) and using the fact that $\phi_{k-j}^k$ is a polynomial of degree $k-j$ then yields, up to a normalization constant,
\[
\phi_{k-j}^k = \pFq{2}{1}{j-k,1-k-j-2 \nu_1 - 2 \nu_2}{1-k-2 \nu_1}{u}.
\]
The same expression was found in \cite{equit} by direct solution of the differential equation.

Note that for $n=3$, Theorem \ref{quantumnrs} reduces to the following spectral equations
\begin{align*}
C_{[2]} \psi_{j}^k &= (j + \nu_1+\nu_2) (j + \nu_1+\nu_2-1) \psi_{j}^k\\
C_{[3]}\psi_{j}^k&= (k + \nu_1+\nu_2+ \nu_3) (j + \nu_1+\nu_2 + \nu_3 -1) \psi_{j}^k.
\end{align*}

Recursive formulas like (\ref{recursiveL}) can be given in general dimension, for the basis of  $\cH_k(\mR^n)$. 
Observe that, using Theorem \ref{basis} and Section \ref{SecVars}
\begin{align*}
\psi_{j_1, \ldots, j_{n-1}} &=  CK_{n} \left(  \left(K_+^{[n-1]}\right)^{j_{n-1}} CK_{n-1} \left(  \ldots \left(K_+^{[3]}\right)^{j_3} CK_{3} \left( \left(K_+^{[2]}\right)^{j_2} CK_{2} \left( x_1^{j_1}   \right) \right) \right) \right)\\
&= (x_1 - x_2)^k \phi_{j_1, \ldots, j_{n-1}}(u_1, \ldots, u_{n-2})
\end{align*}
with $\sum_{\ell=1}^{n-1} j_{\ell} = k$. The initial condition is
\[
 \phi_{j_1, 0, \ldots,0}(u_1, \ldots, u_{n-2}) =1.
\]
Now observe that for $\ell \in \{2, \ldots, n-1 \}$
\[
\psi_{j_1, \ldots, j_{\ell}+1, 0 \ldots, 0} = \left( CK_{\ell+1} K_+^{[\ell]} CK_{\ell+1}^{-1}  \right) \psi_{j_1, \ldots, j_{\ell}, 0 \ldots, 0}
\]
which translates to
\[
\phi_{j_1, \ldots, j_{\ell}+1, 0 \ldots, 0} = L_{u_{1}, \ldots, u_{\ell-1}}^{|j|_{\ell}, \ell} \phi_{j_1, \ldots, j_{\ell}, 0 \ldots, 0}
\]
where
\[
 L_{u_{1}, \ldots, u_{\ell-1}}^{|j|_{\ell}, \ell} =  2 \nu_1+|j|_{\ell}+\sum_{i=1}^{\ell-1}   \left(u_i - 1 + 2 \sum_{p=1}^{i-1}u_p \right) u_i \partial_{u_i} - \sum_{i=1}^{\ell-1} \left(2 |j|_{\ell} + 2 \sum_{p=1}^{i+1} \nu_{p} \right) u_i 
\]
and with $|j|_{\ell} = j_1+ \ldots + j_{\ell}$.
In this way, we can construct $\phi_{j_1, \ldots, j_{n-1}}$ by consecutively adding each index $j_\ell$  for $\ell \in \{2, \ldots, n-1 \}$, and hence also construct $ \psi_{j_1, \ldots, j_{n-1}}$.

\section{The Barut-Girardello model}
\label{BGsec}

Let us now consider another well known realization of $\su$. 
Introduce, for $\nu >0$, the following operators
\begin{align*}
L_0 &= x \partial_x + \nu\\
L_- &= x\partial_x^2 + 2\nu \partial_x\\
L_+ &= x.
\end{align*}
It is easy to verify that they satisfy the $\su$ relations:
\[
[L_0, L_\pm] = \pm L_\pm, \qquad [L_-,L_+] = 2 L_0.
\]
Note that these operators appeared probably first in \cite{BG}, in the context of coherent states.

The connection between the Barut-Girardello model and the Bargmann model can be established using a weighted Laplace transform, see \cite{BVM}. Define
\[
\cL_{\nu}^x (f)(\rho) = \frac{1}{\Gamma{(2\nu)}} \rho^{-2 \nu} \int_0^{+\infty} x^{2\nu-1} f(x) e^{-x/\rho}dx.
\]
Here we have chosen the normalization such that $\cL_{\nu}^x (1)(\rho) =1$.
We can moreover easily compute that
\begin{equation}
\label{laplacemon}
\cL_{\nu}^x (z^n)(\rho) = \frac{\Gamma{(n+ 2\nu)}}{\Gamma{(2\nu)}} \rho^n.
\end{equation}

\begin{remark}
In \cite{H1, H2} Rosengren introduces the inverse of (\ref{laplacemon}). However, he does not give the interpretation as an inverse Laplace transform but only specifies the action on monomials as in (\ref{laplacemon}).
\end{remark}

As a consequence, we have
\begin{equation}
\label{intertwine}
\cL_{\nu}^x ( L_{\pm}f) = K_{\pm}\cL_{\nu}^x (f).
\end{equation}
In other words, the Laplace transform $\cL_{\nu}^x $ intertwines the Barut-Girardello model with the Bargmann one.

\begin{remark}
Note that the algebra  $\mathcal{U}(\mathfrak{su}(1,1))$ has a natural anti-automorphism $\mathfrak{a}$ defined by 
\begin{align*}
\mathfrak{a}(A_0)&=A_0\\
\mathfrak{a}(J_+)&=J_-\\
\mathfrak{a}(J_-)&=J_+,
\end{align*}
using the notations in Section~\ref{prelims}. The Weyl algebra  $\mathbb{C}\langle x,\partial_{x}\rangle$ also has a natural anti-automorphism $\mathfrak{b}$ defined by 
\begin{align*}
\mathfrak{b}(x)&=\partial_x\\
\mathfrak{b}(\partial_x)&=x.
\end{align*}
The composition $\mathfrak{b}\circ\mathfrak{a}$ of these anti-automorphisms provides an isomorphism between 
the Barut-Girardello and Bargmann representations of $\mathfrak{su}(1,1)$.
\end{remark}

Consider now $n$ mutually commuting sets of Barut-Girardello $\su$ generators and combine them by addition. This yields the following operators, defined for any subset $A \subset [n]$
\begin{align*}
L_-^{A} &= \sum_{j \in A} \left(x_j\partial_{x_{j}}^2 + 2 \nu_j \partial_{x_j} \right)\\
L_+^{A} &= \sum_{j \in A}x_j\\
L_0^{A} & = \sum_{j \in A} x_j\partial_{x_{j}} + \sum_{j \in A}\nu_j
\end{align*}
which again generate $\su$. We consider their action on $\cP(\mR^n)= \mR[x_1, \ldots, x_n]$. 
We define the space of Barut-Girardello harmonics as
\[
\cH_k^{BG}(\mR^n) = \cP_k(\mR^n) \cap \ker{L_-^{[n]}}.
\]
If we now put
\[
\cL_{\vec \nu}^{\vec x} = \prod_{j=1}^n \cL_{\nu_j}^{x_j}
\]
then clearly
\[
\cL_{\vec \nu}^{\vec x}  \left(\cH_k^{BG}(\mR^n) \right) = \cH_k(\mR^n).
\]

Surprisingly, we can construct an explicit basis in terms of Jacobi polynomials for $\cH_k^{BG}$, diagonalizing a labelling Abelian subalgebra.  We start with the following two theorems. We omit the proofs, as they are essentially identical to those of the Bargmann case.

\begin{theorem}
\label{FischerBG}
The space $\cP_k(\mR^n)$ decomposes as
\[
\cP_k(\mR^n) = \bigoplus_{j=0}^k \left(L_+^{[n]}\right)^j \cH_{k-j}^{BG}(\mR^n).
\]
\end{theorem}

\begin{theorem}
\label{CKBG}
There exists an isomorphism between $\cP_k(\mR^{n-1})$ and $\cH_k^{BG}(\mR^n)$, given by the map
\[
CK_n^{BG}(p) =  \Gamma(2 \nu_n) \sum_{j=0}^k \frac{(- x_n L_-^{[n-1]})^j}{j ! \Gamma(j + 2 \nu_n)} p, \quad p \in \cP_k(\mR^{n-1}).
\]
\end{theorem}

Note that the explicit form of the CK map is much more complicated in the Barut-Girardello model than in the Bargmann model. However, it is easier to determine the eigenfunctions explicitly as we will show later.
Nevertheless, we already observe that
\[
CK_n(\cL_{\vec \nu}^{\vec x} (p)) = \cL_{\vec \nu}^{\vec x} \left( CK_n^{BG}(p) \right),
\]
thus linking the $CK$ extensions in both frameworks.

We can now obtain a basis for $\cH_k^{BG}(\mR^{n})$, by combining Theorem \ref{CKBG} with \ref{FischerBG}. We have the following result
\begin{theorem}
\label{basisBG}
A basis for $\cH_k^{BG}(\mR^{n})$ is given by the set of polynomials
\[
\psi_{j_1, \ldots, j_{n-1}}^{BG} =  CK_{n}^{BG} \left(  \left(L_+^{[n-1]}\right)^{j_{n-1}} CK_{n-1}^{BG} \left(  \ldots   \left(L_+^{[2]}\right)^{j_2} CK_{2}^{BG} \left( x_1^{j_1}   \right) \right) \right)
\]
with $\sum_{\ell=1}^{n-1} j_{\ell} = k$.
\end{theorem}

Contrary to the Bargmann case, it now becomes relatively easy to express the basis functions $\psi_{j_1, \ldots, j_{n-1}}^{BG} $ explicitly in terms of Jacobi polynomials. This goes as follows.

First recall the following definition of the Jacobi polynomial:
\[
P_n^{(\alpha, \beta)}(x) = \sum_{s=0}^n \binom{n + \alpha}{s}\binom{n + \beta}{n-s} \left( \frac{x-1}{2}\right)^{n-s} \left( \frac{x+1}{2}\right)^{s}
\]
which can be rewritten as
\begin{equation}
\label{jacobiexpr}
n! (u + v)^n P_n^{(\alpha, \beta)}( (v-u)/(u+v)) = \sum_{j=0}^n \binom{n }{j} \frac{\Gamma(n + \alpha+1)\Gamma(n + \beta+1)}{\Gamma(n - j + \alpha+1)\Gamma(j+ \beta+1)}  v^{j} (-u)^{n-j}
\end{equation}
where we put $x = (v-u)/(u+v)$.
Next we observe that, using the $\su$ relations and induction, we have the following formulas for all $p \in \{1, \ldots, n \}$ and $\ell \leq j$
\begin{equation}
\label{suind}
(L_-^{[p]})^{\ell}(L_+^{[p]})^{j} \phi_k = c_{j, k , \ell} (L_+^{[p]})^{j-\ell} \phi_k, \qquad \phi_k \in \cH_k^{BG}(\mR^p)
\end{equation}
with
\[
c_{j, k , \ell} = \frac{j !}{(j- \ell)!} \frac{\Gamma(2 |\nu|_p + 2k + j)}{\Gamma(2 |\nu|_p+2k + j - \ell)}
\]
where
\[
|\nu|_p = \sum_{j=1}^p \nu_j.
\]

Formula (\ref{suind}) combined with Theorems \ref{CKBG} and \ref{basisBG} already implies that the basisfunctions in Theorem \ref{basisBG} factor as follows
\begin{equation}
\label{BGexplicit}
\psi_{j_1, \ldots, j_{n-1}}^{BG} = \prod_{k=1}^{n-1} \phi_{j_k} \left(x_{k+1}, L_+^{[k]}\right).
\end{equation}
These factors can be computed explicitly by expanding the $CK$ extensions, and subsequently using (\ref{suind}) and (\ref{jacobiexpr}). The result is
\[
\phi_{j_k} \left(x_{k+1}, L_+^{[k]}\right) = (-1)^{j_k} (j_k)! \frac{\Gamma(2 \nu_{k+1})}{\Gamma(2 \nu_{k+1} + j_k)} \left(L_+^{[k+1]} \right)^{j_k} P^{(\alpha_k, \beta_k)}_{j_k}\left( \frac{x_{k+1} - L_+^{[k]}}{L_+^{[k+1]}}\right)
\]
with
\begin{align*}
\alpha_k & = 2  |\nu|_k -1 + 2 \sum_{\ell =1}^{k-1}j_\ell \\
\beta_k & = 2 \nu_{k+1} -1
\end{align*}
and where we recall that
\[
L_+^{[k]} = \sum_{\ell =1 }^k x_{\ell}
\]
is nothing but a linear polynomial.

Observe now that
\[
\prod_{j=1}^{\ell} \cL_{\nu_j}^{x_j} \left( CK_\ell^{BG} f \right) = CK_\ell  \left( \prod_{j=1}^{\ell} \cL_{\nu_j}^{x_j} \left( f \right) \right),
\]
which, together with (\ref{intertwine}), implies that
\[
\cL_{\vec \nu}^{\vec x}  \left(\psi_{j_1, \ldots, j_{n-1}}^{BG} \right) = \psi_{j_1, \ldots, j_{n-1}}.
\]
This means that the Laplace transform of the explicit Jacobi basis in the BG model yields an explicit basis in the Bargmann model.

\section{Connection with Miller model}

In this section, we establish the connection between the Barut-Girardello model on the one hand and the Miller model (see formula (\ref{eq:Hamiltonian1})) on the other hand. We will do this both for the eigenfunctions and for the differential equations of the model.

First we make the connection explicit for the eigenfunctions of both models. Start from Theorem \ref{basisBG}, where we apply a permutation to the order in which the CK extensions are applied. This yields the following alternative basis:
\[
\widetilde{\psi_{j_1, \ldots, j_{n-1}}^{BG}} =  CK_{1}^{BG} \left(  \left(L_+^{\{2, \ldots, n\}}\right)^{j_{1}} CK_{2}^{BG} \left(  \ldots   \left(L_+^{\{n-1,n\}}\right)^{j_{n-2}} CK_{n-1}^{BG} \left( x_n^{j_{n-1}}   \right) \right) \right)
\]
which factors as
\[
\widetilde{\psi_{j_1, \ldots, j_{n-1}}^{BG}} = \prod_{k=1}^{n-1} \phi_{j_k} \left(x_{k}, L_+^{\{k+1, \ldots, n\}}\right).
\]
If we now restrict to the hyperplane $x_1+ \ldots + x_n =1$, then after some computations $\phi_{j_k}$ is given explicitly by
\[
\phi_{j_k} = c \left(1 - \sum_{i=1}^{k-1} x_i \right)^{j_k} P^{(\alpha_k, \beta_k)}_{j_k}\left( \frac{2 x_k}{1 - \sum_{i=1}^{k-1} x_i} -1\right)
\]
with $c$  a constant and
\begin{align*}
\alpha_k & = -1 +2 \sum_{\ell =k+1}^{n-1}j_\ell+ 2 \sum_{\ell =k+1}^{n}\nu_\ell \\
\beta_k & = 2 \nu_{k} -1.
\end{align*}
This is precisely the form of the eigenfunctions in the Miller model as given in \cite{I, I2}.

Second, consider the equation
\[
\sum_{j =1}^n \left(x_j\partial_{x_{j}}^2 + 2 \nu_j \partial_{x_j} \right) f(x_1, x_2, \ldots, x_n) =0.
\]
Under the change of variables $x_i = z_i^2$ this equation becomes
\[
\frac{1}{4} \sum_{j =1}^n \left(  \partial_{z_j}^2 +  \frac{4 \nu_j - 1}{z_j} \partial_{z_j}  \right) f(z_1, z_2, \ldots, z_n) =0.
\]
The terms in $\partial_{z_i}$ can be gauged away. Indeed, if
\[
H = \sum_{j =1}^n \left( \partial_{z_j}^2 + \frac{4 \nu_j - 1}{z_j} \partial_{z_j}  \right)
\]
then $\widetilde{H} = g H g^{-1}$ with
\[
g = \prod_{j=1}^n z_j^{2 \left(\nu_j-1/4 \right)}
\]
is given by
\[
\widetilde{H} =  \Delta_z - \sum_{j =1}^n \left((2\nu_j-1)^2-\frac{1}{4} \right) \frac{1}{z_j^2}.
\]
Here $\Delta_z$ is the Laplace operator in $(z_1, \ldots, z_n)$. Finally, we restrict $\widetilde{H}$ to the unit sphere $z_1^2 + \ldots + z_n^2=1$. Then we obtain
\[
\widetilde{H}_{\mS^{n-1}} = \sum_{1\leq i < j \leq n} (z_j \partial_{z_i} - z_i \partial_{z_j})^2
 -\sum_{j=1}^{n}\frac{\left((2\nu_j-1)^2-\frac{1}{4} \right)}{z_j^2},
\]
which is indeed the superintegrable Hamiltonian (\ref{eq:Hamiltonian1}) of Miller et al.

\section{Conclusions and outlook}

In this paper we introduced two new models in which the higher rank Racah algebra naturally appears. They are based on the Bargmann and Barut-Girardello realizations of $\su$. We have shown that both models are isomorphic through a modified Laplace transform and have determined eigenfunctions and eigenvalues for a labelling Abelian subalgebra. We have moreover identified the connection with the superintegrable model of Miller et al. in two different ways. Finally, our approach has yielded a differential operator realization of the rank $n$ Racah algebra in precisely $n$ variables. This opens the door for identifying the Racah algebra in new ways and we plan to report on that in the future.
Another important problem is to extend the models introduced here to models of the Bannai-Ito algebra (see \cite{DBAdv}) and to find a generalization of the embedding of \cite{Bas}.

 \section*{Acknowledgements}
The work of HDB is supported by the Research Foundation Flanders (FWO) under Grant EOS 30889451. HDB and PI are grateful for the hospitality extended to them by the Centre de Recherches Math\'ematiques in Montr\'eal, where most of this research was carried out. The research of LV is funded in part by a discovery grant of the Natural Sciences and Engineering Council (NSERC) of Canada.


\end{document}